\documentclass{llncs}

\usepackage[pdftex]{graphicx}
\usepackage{url}
\usepackage{ascmac,bm}
\usepackage{amsmath,amssymb}
\usepackage{float}
\usepackage{algorithm,algorithmic}
\usepackage{extarrows}

\spnewtheorem*{lem}{Lemma}{\bfseries}{\rmfamily}
\spnewtheorem*{defn}{Definition}{\bfseries}{\rmfamily}
\spnewtheorem*{prop}{Proposition}{\bfseries}{\rmfamily}
\spnewtheorem*{thm}{Theorem}{\bfseries}{\rmfamily}

\newcommand{\ala}{{\it \`a la} }
\newcommand{\bala}{{\it \textbf{\`a la}} }

\newcommand{\ie}{{\it i.e.}}

\newcommand{\nat}{\mathbb{N}}

\newcommand{\hole}{[\,]}
\newcommand{\yield}{Y}
\newcommand{\nodes}{N}
\renewcommand{\root}{R}

\newcommand{\decomps}{\Delta}
\newcommand{\decomp}[1]{\decomps\!\left(#1\right)}
\newcommand{\adj}[1]{\mathrm{Adj}^*\!\!\left(#1\right)}
\newcommand{\Adj}[1]{\mathrm{Adj}^+\!\!\left(#1\right)}
\newcommand{\defeq}{\triangleq}

\newcommand{\vatree}{{\((V, A)\)-tree}}

\newcommand{\trees}[1]{{\cal T}\!\left(#1\right)}
\newcommand{\lang}[1]{{\cal L}\!\left(#1\right)}
\newcommand{\strees}[1]{{\cal S}\!\left(#1\right)}
\newcommand{\sloops}[1]{{\cal A}\!\left(#1\right)}

\begin{document}
\title{Simple proof of Parikh's theorem \bala  Takahashi}
\author{Ryoma Sin'ya}
\institute{Akita University\\
\email{ryoma@math.akita-u.ac.jp}}
\maketitle

\begin{abstract}
In this report we describe a simple proof of Parikh's theorem \ala Takahashi, based on a decomposition of derivation trees.
The idea of decomposition is appeared in her master's thesis written in 1970.
\end{abstract}

%\keywords{regular languages, automata, spectra of graphs, partition}
%\section{Prologue}
\section{Preliminaries}
For a set \(S\), we denote by \(|S|\) the cardinality of \(S\). The set of
natural numbers including \(0\) is denoted by \(\nat\).
Let \(G = (V, D, X_0)\) be a context-free grammar over an alphabet
\(A\) where \(V \, (V \cap A = \emptyset)\) is a finite set of
\emph{non-terminals}, \(D \subseteq V \times (V \cup A \cup \{\epsilon\}
)^+\) is a finite set of \emph{derivation rules}, and \(X_0 \in V\).
The set of \emph{{\vatree}s}, ranged over by \(T\), is given by the
following grammar:
\[
 T ::= a \,\,\, (a \in A \cup \{\epsilon\}) \mid X(T_1, \cdots, T_n) \,\,\, (X \in V, n \geq 1)
\]
Namely, {\vatree}s are trees whose internal nodes are non-terminals, and
whose leaves are letters in \(A\) or the special symbol \(\epsilon
\notin A\).
For a {\vatree} \(T\), we denote by \(\nodes(T)\)
the set of all non-terminals appeared in \(T\),
and denote by \(\root(T)\) the root of \(T\).
The \emph{yield} \(Y\) is a function from
{\vatree}s into \(A^*\) defined inductively as
$Y(a) = a, Y(\epsilon) = \varepsilon$ where \(\varepsilon\) is the empty
string, and \(Y(X(T_1, \ldots, T_n))
 = Y(T_1) \cdots Y(T_n)\).
We call a \((V, A \cup \{\hole\})\)-tree
\(C\) \emph{context} if
exactly one leaf of \(C\) is the special symbol \(\hole \notin A\).% called \emph{hole}.
We denote by \(C[T]\) the {\vatree} obtained by replacing \(\hole\) in
\(C\) by \(T\).
We define \emph{the set \(\trees{G}\) of derivation trees of \(G\)} as
\begin{align*}
 \trees{G} \defeq \{ T\!:\, & \text{{\vatree}} \mid 
 \root(T) = X_0, \text{ for each context } C,\\
& \qquad
 T = C[X(T_1, \ldots, T_n)] \text{ implies }
  (X, \root(T_1) \cdots \root(T_n)) \in D\}
\end{align*}
and define \(\lang{G} \defeq \{ Y(T) \mid T \in \trees{G}\}\).
 
For a non-terminal \(X \in V\), we call a \((V, A \cup \{X\})\)-tree
\(\alpha \neq X\)
an \emph{adjunct tree}
if \(\root(\alpha) = X\) and exactly one leaf of \(\alpha\) is \(X\).
For a \vatree{} \(T\) and an adjunct tree \(\alpha\)
such that \(\root(T) = \root(\alpha)\),
we denote by \(\alpha[T]\) the
{\vatree} obtained by replacing the leaf \(X\) in \(\alpha\) by \(T\).
For a \vatree{} \(T\) and an adjunct tree \(\alpha\), if the root \(X\)
of \(\alpha\) is appeared in \(T\), \ie,
\(T = C[X(T_1, \ldots, T_n)]\) for some context \(C\) and {\vatree}s
\(T_1, \ldots, T_n\),
we say that \emph{\(\alpha\) is adjoinable to \(T\)}, and we say that
\emph{\(T' = C[\alpha[X(T_1, \ldots, T_n)]]\) is obtained from \(T\)
adjoining \(\alpha\)}
and write \(T \vdash_\alpha T'\).
Intuitively, an adjunct tree represents ``pump'' part,
and adjoining corresponds to ``pumping'' operation for trees.
For example,
\(X(Y(X), a)\) is adjoinable to \(Z(X(b))\) and we have
\(Z(X(b)) \vdash_{X(Y(X), a)} Z(X(Y(X(b)), a))\).

We call a {\vatree} \(T\) \emph{simple} if,
for any path in \(T\) from the root to a leaf,
no non-terminal appears more than once.
We call an adjunct tree \(\alpha\) \emph{simple} if,
for any path in \(T\) from a child of the root to a leaf, no
non-terminal appears more than once.
See Fig.~\ref{fig:trees} for example.
\(T_1\) is simple since all paths \(\{(Z, X, a), (Z, X, b)\}\) contain \(Z\)
and \(X\) exactly once.
\(T_2\) is not simple since the left-most path
\((X, Z, X, a)\) contains \(X\) twice.
However, the adjunct tree \(\alpha_1\), which is obtained by removing
the left-most leave \(a\) from \(T_2\) (\ie, \(X(a) \vdash_{\alpha_1}
T_2\)),
is simple since all paths from a child of the root to a leaf \(\{(Z,X), (Z,X,b), (Y,a), (Y,X,b)\}\)
contain no non-terminal more than once.

For a {\vatree} \(T\) and a set of adjunct
trees \(S\), we define
\begin{align*}
  \adj{T, S} &\defeq \{ T' \mid T = T_0 \vdash_{\alpha_1}
  T_1 \vdash_{\alpha_2} \cdots \vdash_{\alpha_k} T_k = T', k \in \nat,
 \{\alpha_1, \ldots, \alpha_k\} \subseteq S \}\\
 \Adj{T, S} &\defeq \{ T' \mid T = T_0 \vdash_{\alpha_1}
  T_1 \vdash_{\alpha_2} \cdots \vdash_{\alpha_k} T_k = T', k \in \nat,
 \{ \alpha_1, \ldots, \alpha_k \} = S \}
\end{align*}
Intuitively, \(\adj{T, S}\)
(resp. \(\Adj{T, S}\))
is the set of all {\vatree}s obtained from \(T\) adjoining each element
in \(S\) arbitrary number of times (resp. arbitrary \emph{positive} number of times).
Clearly,
\(\adj{T, S} = \bigcup_{U \subseteq S} \Adj{T, U}\)
and \(\Adj{T, \emptyset} = \{T\}\).
We say that \emph{\(S\) is adjoinable to \(T\)} if
\(\Adj{T, S}\) is non-empty. Notice that
if \(\Adj{T, S}\) is non-empty then
there exists \(T' \in \Adj{T, S}\) such that \(T'\)
 is obtained from \(T\) adjoining each element in \(S\) \emph{exactly
 once}, \ie,
 \(T_0 = T \vdash_{\alpha_1} T_1 \vdash_{\alpha_2} \cdots
 \vdash_{\alpha_{|S|}} T_{|S|} = T'\)
 and \(S = \{ \alpha_1, \ldots, \alpha_{|S|}\}\).
Moreover, such \(T' \in \Adj{T, S}\) contains every
 root non-terminal of \(\alpha \in S\),
 \(\Adj{T', S}\) is also non-empty
 and thus \(\Adj{T, S}\) should be infinite (if \(S\) is non-empty).
\begin{figure}[t]
\centering{\includegraphics[width=0.65\columnwidth]{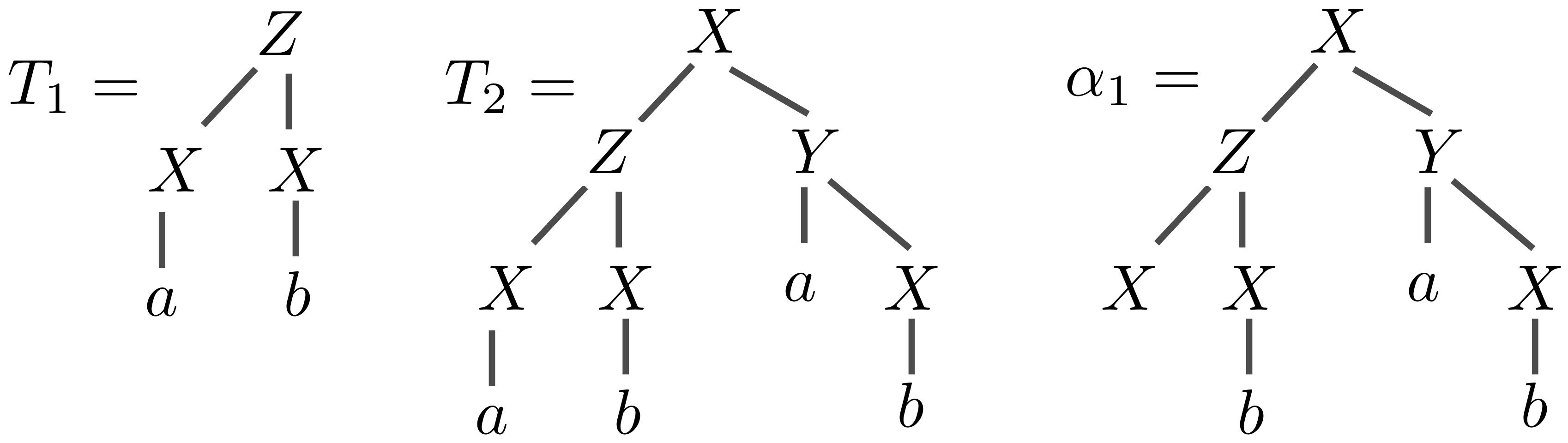}}
\caption{Example of simple \vatree{} \(T_1\), non-simple \vatree{} \(T_2\), and simple adjunct tree \(\alpha_1\)}
\label{fig:trees}
\end{figure}

Let \(A = \{a_1, \cdots, a_d\}\).
The \emph{Parikh mapping} \(\Phi_A: A^* \rightarrow \nat^{d}\) is defined by
\(\Phi_A(w) \defeq (|w|_{a_1}, \ldots, |w|_{a_d})\) where \(|w|_a\)
denotes the number of occurrences of \(a\) in \(w\).
For a {\vatree} \(T\) and an adjunct tree \(\alpha\) where \(X = \root(\alpha)\), we can naturally
extend the definition of the Parikh mapping as
\(\Phi_A(T) \defeq \Phi_A(\yield(T))\) and
\(\Phi_A(\alpha) \defeq \Phi_A(\yield(\alpha[X(\epsilon)]))\).
By definition, we have \(\Phi_A(\lang{G}) = \Phi_A(\trees{G})\) for any
context-free grammar \(G\).
A set \(S \subseteq \nat^d\) is called \emph{linear}
if \(S\) is of the form
\[
 S = \{\bm{v_0} + x_1 \bm{v_i} + \cdots + x_k \bm{v_k} \mid x_i \in \nat \text{ for
 each } i\}
 \] for some \(k \in \nat\) and some vectors
 \(\bm{v_0}, \bm{v_1}, \ldots, \bm{v_k} \in \nat^d\), and we call
a finite union of linear sets \emph{semilinear}.
\section{Proof \bala Takahashi}
\begin{defn}[decomposition]
A decomposition \(\decomp{T}\)
of a {\vatree} \(T\) is defined inductively as follows.
 If \(T = a \in A \cup \{\epsilon\}\), define $\decomp{T} \defeq (a, \emptyset)$.
If \(T = X(T_1, \ldots, T_n)\), let \((T'_1, S_1) = \decomp{T_1}, \ldots, (T'_n, S_n) = \decomp{T_n}\) and define
\begin{align*}
\decomp{T} \defeq \begin{cases}
			       (X(T'_1, \ldots, T'_n), S_1 \cup \cdots \cup S_n) & X \not\in
			       \nodes(T'_1) \cup \cdots \cup \nodes(T'_n)\\
			       (T', \{\alpha\} \cup S_1 \cup \cdots \cup S_n) & X \in \nodes(T'_1)
				       \cup \cdots \cup \nodes(T'_n)
				      \end{cases}
\end{align*}
where
 \(T'\) is the left-most \(X\)-rooted proper subtree of \(X(T'_1,
 \ldots, T'_n)\), \ie, the left-most \(X\)-rooted subtree of \(T'_i\)
 (where \(X \in \nodes(T'_i)\) and \(X \notin \nodes(T'_j)\) for each
 \(1 \leq j < i\)), and \(\alpha\) is the adjunct tree obtained by
 replacing \(T'\) by \(X\) in \(X(T'_1, \ldots, T'_n)\).
\end{defn}
See Fig.~\ref{fig:trees} for example. The non-simple tree \(T_2\) is decomposed as \(\decomp{T_2} = (X(a), \{\alpha_1\})\); it is clear that
\(X(a)\) is the left-most \(X\)-rooted proper subtree of \(T_2\) and \(X(a) \vdash_{\alpha_1} T_2\).

Let \(G = (V, D, X_0)\) be a context-free grammar over \(A\).
\begin{lem}\label{lemma}
For any \(T \in \trees{G}\) and \((T', S) = \decomp{T}\), 
 (1) \(T'\) is simple and \(T' \in \trees{G}\),
 (2) \(S\) is a set of simple adjunct trees, and
 (3) \(T \in \Adj{T', S} \subseteq \trees{G}\).
\end{lem}
\begin{proof}
Straightforward induction on \(T\).
\end{proof}

We define \(\strees{G} \defeq \{ T' \mid (T', S) =
\decomp{T} \text{ for some } T \in \trees{G}
 \text{ and } S \}\) and define
 \(\sloops{G} \defeq \{ \alpha \in S \mid
 (T', S) = \decomp{T}
 \text{ for some } T \in \trees{G}
 \text{ and } S \}\).
Because there are only finitely many simple {\vatree}s (resp.
simple adjunct trees),
\(\strees{G}\) and \(\sloops{G}\) are both finite by Claim (1)--(2) of
Lemma.
 
\begin{prop}[Takahashi \cite{takahashi}]
$\trees{G} = \bigcup_{T \in \strees{G}} \adj{T, \sloops{G}}$.
\end{prop}
\begin{proof}
Left-to-right inclusion \(\subseteq\) is clear by Lemma.
Right-to-left inclusion \(\supseteq\) is shown by induction.
The base case \(T' \in \strees{G} \subseteq \trees{G}\) is trivial.
 Assume \(T' \in \trees{G}\). Then for any
 \(\alpha \in \sloops{G}\)
 such that \(\alpha\) is adjoinable to \(T'\), since \(\alpha\) is
 extracted from some valid derivation tree in \(\trees{G}\), \(T'
 \vdash_{\alpha} T''\) is also in \(\trees{G}\).
\end{proof}

\begin{thm}[Parikh \cite{parikh}]
\(\Phi_A(\lang{G})\) is semilinear.
\end{thm}
\begin{proof}
\begin{align*}
 \Phi_A(\lang{G}) = \Phi_A(\trees{G}) =
\bigcup_{T \in \strees{G}}
\bigcup_{S \subseteq \sloops{G}} \Phi_A(\Adj{T, S})
\end{align*}
holds by Proposition.
If \(S\) is not adjoinable to \(T\)
 then \(\Phi_A(\Adj{T, S}) =
 \emptyset\).
 Otherwise, 
 \(\Phi_A(\Adj{T, S}) = \{ \Phi_A(T) + \sum_{i = 1}^{|S|} x_i
 \Phi_A(\alpha_i) \mid
 S = \{\alpha_1, \ldots, \alpha_{|S|}\},
 x_i \in \nat \setminus \! \{0\}\}\) holds since
 \(T' \vdash_{\alpha} T''\) implies
 \(\Phi_A(T'') = \Phi_A(T') + \Phi_A(\alpha)\).
 In both cases, \(\Phi_A(\Adj{T, S})\) is semilinear,
 hence those finite union \(\Phi_A(\lang{G})\) is semilinear.
\end{proof}

% \noindent
% {\bf Historical remark:}
% three-line\\
% short-remark\\
% is available
%\section{Epilogue}

\bibliographystyle{splncs.bst}
\bibliography{ref}

\end{document}